\documentclass[11pt]{article}
\usepackage{enumerate, enumitem, hyperref}
\usepackage{latexsym}
\usepackage{amsmath}
\usepackage{amssymb}
\usepackage{amsthm}
\usepackage{epsfig}
\usepackage{bbm}
\usepackage[latin1]{inputenc}
\usepackage{color}

\def\namedlabel#1#2{\begingroup
    #2%
    \def\@currentlabel{#2}%
    \phantomsection\label{#1}\endgroup
}

\usepackage{tikz}
\usetikzlibrary{shapes}
\usetikzlibrary{patterns}
\usetikzlibrary{arrows,positioning,decorations.pathmorphing,trees}

\usepackage{fullpage}

\newtheorem{theorem}{Theorem}[section]
\newtheorem{claim}{Claim}[section]
\newtheorem{lemma}{Lemma}[section]

\theoremstyle{remark}

\newenvironment{restatetheorem}[1]
{\begingroup
  \def\thetheorem{\ref{#1}}
  \begin{theorem}}

\newtheorem{innercustomgeneric}{\customgenericname}
\providecommand{\customgenericname}{}
\newcommand{\newcustomtheorem}[2]{%
  \newenvironment{#1}[1]
  {%
   \renewcommand\customgenericname{#2}%
   \renewcommand\theinnercustomgeneric{##1}%
   \innercustomgeneric
  }
  {\endinnercustomgeneric}
}

\newcustomtheorem{customthm}{Theorem}
\newcustomtheorem{customlemma}{Lemma}

\usepackage[noend]{algorithmic}
\usepackage{algorithm}
\floatname{algorithm}{Procedure}

\title{The Blocker Postulates for Measures of Voting Power}
\author{A. Abizadeh\thanks{Department of Political Science, McGill University: {\tt arash.abizadeh@mcgill.ca}} 
\and A. Vetta\thanks{Department of Mathematics \& Statistics and School of Computer Science, McGill University: {\tt adrian.vetta@mcgill.ca}}}

\begin{document}

\maketitle

\begin{abstract}
A proposed measure of voting power should satisfy two conditions to be plausible:
first, it must be conceptually justified, capturing the intuitive meaning of what voting power is;
second, it must satisfy reasonable postulates.
This paper studies a set of postulates, appropriate for a priori voting power,
concerning blockers (or vetoers) in a binary voting game.
We specify and motivate five such postulates, namely,
two {\em subadditivity blocker postulates}, two {\em minimum-power blocker postulates}, each in weak and strong versions,
and the {\em added-blocker postulate}.
We then test whether three measures of voting power,
namely the classic Penrose-Banzhaf measure, the classic Shapley-Shubik index,
and the newly proposed Recursive Measure, satisfy these postulates.
We find that the first measure fails four of the postulates,
the second fails two,
while the third alone satisfies all five postulates.
This work consequently adds to the plausibility of the Recursive Measure as a reasonable measure of voting power.

\end{abstract}

\section{Introduction}\label{sec:intro}

A proposed measure of voting power should satisfy two conditions to be plausible:
first, it must be conceptually justified, in the sense that it captures the intuitive meaning of what voting power is;
second, it must satisfy reasonable postulates for measures of voting power.
Numerous postulates have been defended in the voting-power literature: most are for {\em a priori} voting power
(i.e., voting power solely in virtue of the formal voting structure itself, constituted by the agenda of potential outcomes, the sets of actors, 
their action profiles, and the decision function mapping vote configurations onto outcomes), while others are for {\em a posteriori} voting power
(i.e., voting power also in virtue of the distribution of preferences, and consequent incentives for strategic interaction, within the voting structure)
(Felsenthal and Machover 1998; Laruelle and Valenciano 2005a).

This paper studies a set of postulates, appropriate for a priori voting power, concerning blockers (or vetoers) in a binary voting game.
Our aim is two-fold. First, to specify and motivate five postulates concerning blockers, namely,
two {\em subadditivity blocker postulates}, two {\em minimum-power blocker postulates}, each in weak and strong versions,
and the {\em added-blocker postulate}.
Second, to test whether three measures of voting power,
namely the classic Penrose-Banzhaf measure (PB), the classic Shapley-Shubik index (SS),
and the newly proposed Recursive Measure (RM), satisfy these postulates.
We find that PB fails the first four postulates,
SS fails the strong subadditivity blocker postulate and the added-blocker postulate,
while RM alone satisfies all five postulates.
Further, it is already known (Abizadeh and Vetta 2021) that RM satisfies a plethora of other 
reasonable postulates, including the iso-invariance, dummy, dominance, donation, minimum-power bloc, and quarrel postulates.
This work consequently adds to the plausibility of RM as a reasonable measure of voting power.

\section{Three Measures of Voting Power}\label{sec:MVP}

In this section, we present the three measures of voting power studied in this paper.
The first are the two classic measures of voting power, namely, the Penrose-Banzhaf measure and Shapley-Shubik index.
The third is the aforementioned Recursive Measure proposed by Abizadeh and Vetta.
Before formally defining these three measures, we introduce the notions of a simple voting game and a measure of voting power.

\subsection{Simple Voting Games}\label{sec:SVG}

Here we present the class of voting games, called {\em simple voting games} (SVGs),
for which it would be reasonable to expect measures of voting power
to satisfy our postulates.
Denote by $[n]=\{1,2,\dots, n\}$ a nonempty, finite set of players with $z=2$ strategies, voting {\sc yes} or voting {\sc no}.
Let $\mathcal{O}$=\{{\sc yes}, {\sc no}\} be the set of alternative outcomes.
A {\em division} or complete vote configuration $\mathbb{S}=(S, \bar{S})$ of the set $[n]$ is an ordered partition of players where 
the first element in the 
ordered pair is the set of {\sc yes}-voters and the second element is the set of {\sc no}-voters in $\mathbb{S}$.
Thus, for $\mathbb{S}=(S, \bar{S})$, the subset $S\subseteq [n]$ comprises the set of {\sc yes}-voters 
and the subset $\bar{S}= [n]\setminus S$ comprises the set of {\sc no}-voters.
Note the convention of representing a bipartitioned division by its first element in blackboard bold.

Let $\mathcal{D}$ be the set of all logically possible divisions $\mathbb{S}$ of $[n]$.
A {\em binary voting game}, in which each player has two possible strategies,
is a function $\mathcal{G}(\mathbb{S})$ mapping the set of all possible divisions $\mathcal{D}$ to the two outcomes in $\mathcal{O}$.
A {\em monotonic} binary voting game is one satisfying the condition:\\
\indent (i) {\tt Monotonicity.} If $\mathcal{G}$($\mathbb{S}$)={\sc yes} and $S\subseteq T$, then $\mathcal{G}(\mathbb{T})$={\sc yes}.\\
A SVG is a monotonic binary voting game that also satisfies {\em non-triviality}:\\
\indent (ii) {\tt Non-Triviality.} $\exists \mathbb{S}$ $\mathcal |\ {G}$($\mathbb{S}$)={\sc yes} 
and $\exists \mathbb{T}$ $\mathcal |\ {G}$($\mathbb{T}$)={\sc no}.\\
Monotonicity and non-triviality jointly ensure that SVGs also have the following property:\\
\indent (iii) {\tt Unanimity.} $\mathcal{G}((\emptyset, [n]))$={\sc no} and $\mathcal{G}(([n], \emptyset))$={\sc yes}.\\
Note that {\em unanimity} itself implies non-triviality. Thus conditions (i) and (iii) also characterize the class of SVGs.

Call any player whose vote corresponds to the division outcome a {\em successful} player.
Let $\mathcal{W}$ be the collection of all sets of players $S$ such that $\mathcal{G}$($\mathbb{S}$)={\sc yes} (that is, if each 
member of $S$ were to vote {\sc yes}, they would be successful {\sc yes}-voters).
We call this the collection of {\em {\sc yes}-successful subsets} of $[n]$, also commonly called {\em winning coalitions}.
We can now alternatively characterize conditions (i)-(iii) as:\\
\indent (i) {\tt Monotonicity.} If $S\in \mathcal{W}$ and $S\subseteq T$ then $T\in \mathcal{W}$.\\
\indent (ii) {\tt Non-Triviality.} $\exists S |$ $S \in \mathcal{W}$ and $\exists T |$ $T \notin \mathcal{W}$\\
\indent (iii) {\tt Unanimity.} $[n]\in \mathcal{W}$ and $\emptyset \notin \mathcal{W}$.\\
In the discussion and proofs that follow, as is standard in the voting-power literature, we assume that our voting games are SVGs.

\subsection{Voting Power}\label{sec:voting-power}
We define a {\em measure of voting power} for SVGs as a function $\Psi$ that assigns to each player $i$ a nonnegative real number $\Psi_i \geq 0$
and that satisfies two sets of basic adequacy postulates:
the {\em iso-invariance} postulate, according to which the a priori voting power of any player according to that measure remains the same between two isomorphic games;
and the {\em dummy} postulates, according to which a player has zero a priori voting power if and only if it is a 
dummy (i.e., it is not decisive in any division),
and the addition of a dummy to a voting structure leaves other players' a priori voting power unchanged (Felsenthal and Machover 1998: 236).

A measure of voting power can be represented as assigning to each player $i$ a value
\begin{align*}
\Psi_i &= \sum_{\mathbb{S}\in \mathcal{D}} \alpha_i(\mathbb{S})\cdot \gamma(\mathbb{S})
\end{align*}
where $\alpha_i(\mathbb{S})$ is the {\em division efficacy score} of player $i$ in division $\mathbb{S}$
and $\gamma(\mathbb{S})$ is the {\em division weight} assigned to $\mathbb{S}$ for any division $\mathbb{S}\in \mathcal{D}$.
The defining characteristic of a given measure of voting power is therefore its specification of a player's division efficacy score for each division
and each division's weight.
We shall label the {\em a priori} voting power of a player according to a measure~$\Psi$,
i.e., in abstraction from any information about the distribution of preferences,
using the lower case~$\psi$.

Before proceeding, we introduce one further set of concepts.
Let a player $i$'s {\em {\sc yes}-efficacy score} $\alpha_i^+$ be equal to $\alpha_i$ in divisions in which $i$ votes {\sc yes}, equal to 0 otherwise;
and $i$'s {\em {\sc no}-efficacy score} $\alpha_i^-$ be equal to $\alpha_i$ in divisions in which $i$ votes {\sc no}, equal to 0 otherwise.
We say that a player's {\em {\sc yes}-voting power} $\Psi^+$ sums over its weighted {\sc yes}-efficacy scores,
and its {\em {\sc no}-voting power} $\Psi^-$ sums over its weighted {\sc no}-efficacy scores. That is, for SVGs,
\begin{align*}
\Psi_i^+ = \sum_{\mathbb{S}\in \mathcal{D}} \alpha_i^+(\mathbb{S})\cdot \gamma(\mathbb{S}) \qquad\qquad
\Psi_i^- = \sum_{\mathbb{S}\in \mathcal{D}} \alpha_i^-(\mathbb{S})\cdot \gamma(\mathbb{S}) \qquad\qquad
\Psi_i = \Psi_i^+ + \Psi_i^-
\end{align*}

We may now present the three aforementioned measures of voting power.

\subsection{The Penrose-Banzhaf Measure}\label{sec:PB}

The Penrose-Banzhaf (PB) measure bases a voter's division efficacy score $\alpha_i(\mathbb{S})$ on being {\em decisive},
i.e., being in a position in which one could have effected a different outcome by (unilaterally) voting differently than one did in a given division.
The concept can be formalized for SVGs as follows:
a player $i$ is {\sc yes}{\em -decisive} in division $\mathbb{S}$ if and only if $i\in S\in \mathcal{W}$ but $S\setminus\{i\}\notin \mathcal{W}$;
is {\sc no}{\em -decisive} if and only if $i\notin S\notin \mathcal{W}$ but $S\cup\{i\}\in \mathcal{W}$;
and is {\em decisive} if and only if it is either {\sc yes}-decisive or {\sc no}-decisive.
PB, which was originally conceived as a measure of a priori voting power,
then equates a player's voting power with the proportion of logically possible divisions in which it is decisive.
This, in turn, is typically taken to represent the ex ante {\em probability} that a player $i$ will be decisive in a voting structure,
under the assumptions of {\em voting independence} (votes are not correlated)
and {\em equiprobable voting} (the probability a player votes for one alternative equals the probability it votes for any other),
which together imply {\em equiprobable divisions} --
which assumptions model the a prioristic abstraction from voter preferences. 

In particular, PB is defined by specifying a voter's division efficacy score as
$$
\alpha^{PB}_i(\mathbb{S}) =
\begin{cases}
1& \mathrm{if\ } i \mathrm{\ is\ decisive\ in\ } \mathbb{S}\\
0 & \mathrm{otherwise}
\end{cases}
$$

\noindent A division's weight $\gamma(\mathbb{S})$, in turn, is typically interpreted in the Penrose-Banzhaff model as $\mathbb{S}$'s 
ex ante {\em probability}.
In the general, a posteriori case, a division's probability $\mathbb{P}(\mathbb{S})$ would be a function of the actual distribution of voter preferences;
but a division's probability $\mathbbm{p}(\mathbb{S})$ in the a priori case
(which we represent again using the lower case) assumes equiprobable divisions.
For binary voting games ($z=2$), this yields:

$$
\gamma^{PB}(\mathbb{S}) = \mathbbm{p}(\mathbb{S}) = \frac{1}{|\mathcal{D}|} = \frac{1}{z^n} = \frac{1}{2^n}
$$

Notice that, because PB calculates a voter's division efficacy score strictly on the basis of whether the voter is decisive in that division,
its efficacy scores are what we shall call {\em strategy symmetric}, that is,
a player's efficacy score in a given division is equal to its efficacy score in any other division that is identical to it but for the player's own vote.
In the case of SVGs, where $z=2$,
this is because for every division in which the voter is {\sc yes}-decisive
there is precisely one corresponding division in which the voter is {\sc no}-decisive
(involving the two divisions that are identical except for the vote of the player in question).
It follows that in SVGs the number of divisions in which the player plays a given strategy
and is decisive equals the number of divisions in which it plays any other strategy and is decisive.
Strategy symmetry implies that $PB_i^+=PB_i^-$.

It follows that PB can be calculated via a shortcut $PB^*$, on the basis of solely {\sc yes}-decisiveness
(indeed, PB is typically {\em defined} in this way in the literature),
setting a player's division efficacy score as:

$$
\alpha^{PB^*}_i(\mathbb{S}) = \alpha^{+PB}_i(\mathbb{S}) =
\begin{cases}
1& \mathrm{if\ } i \mathrm{\ is\ {\text{{\sc yes}-decisive}} \ in\ } \mathbb{S}\\
0 & \mathrm{otherwise}
\end{cases}
$$
and each division's weight (again, for a priori power) as:
$$
\gamma^{PB^*}(\mathbb{S}) =  \frac{1}{|\mathcal{D}|/z} = \frac{1}{z^{n-1}} = \frac{1}{2^{n-1}}
$$
\noindent Precisely because PB is strategy symmetric, we could also construct a corresponding shortcut based on {\sc no}-decisiveness.

\subsection{The Shapley-Shubik Index}\label{sec:SS}

When Shapley and Shubik (1954) initially introduced their index,
they characterized a player's a priori voting power as equal to the proportion of permutations (ordered sequences) of voters in which a 
voter would be {\em pivotal},
i.e., the probability that the player would be pivotal if all permutations of voters are equiprobable.
A pivotal voter is one who, in an ordered sequence of voters who sequentially vote in favour of an alternative,
is the first whose vote secures it regardless of how subsequent voters vote.
Subsequent analysis has shown that being pivotal is analytically reducible to the notion of decisiveness (Turnovec et al. 2008),
which is why we can define Shapley-Shubik (SS) index using our general formula for $\Psi$ above.
We begin by specifying a player's division efficacy score, as with PB, as follows:

$$
\alpha^{SS}_i(\mathbb{S}) =
\begin{cases}
1& \mathrm{if\ } i \mathrm{\ is \ decisive \ in\ } \mathbb{S}\\
0 & \mathrm{otherwise}
\end{cases}
$$
\noindent and then set each division's weight for SVGs, where $k$ equals the number of voters whose vote agrees with $i$'s vote 
in division $\mathbb{S}$, as:
$$
\gamma^{SS}(\mathbb{S}) = \frac{(k-1)!\cdot (n-k)!}{2n!}
$$

Because SS, like PB, is strategy symmetric, it too can be calculated via a shortcut $SS^*$, on the basis solely 
of {\sc yes}-decisiveness; indeed, this is how SS is typically defined in the literature (e.g. Felsenthal and Machover 1998):

$$
\alpha^{SS^*}_i(\mathbb{S}) = \alpha^{+SS}_i(\mathbb{S}) =
\begin{cases}
1& \mathrm{if\ } i \mathrm{\ is\ {\text{{\sc yes}-decisive}} \ in\ } \mathbb{S}\\
0 & \mathrm{otherwise}
\end{cases}
$$

\noindent and:

$$
\gamma^{SS^*}(\mathbb{S}) = \frac{(|S|-1)!\cdot (n-|S|)!}{n!}
$$
\noindent As with PB, we could also construct the corresponding shortcut via {\sc no}-decisiveness.

\subsection{The Recursive Measure}\label{sec:RM}

The Recursive Measure (RM) specifies division efficacy scores in terms of not just the voter's decisiveness in the division,
but recursively in terms of its {\em degree of efficacy} in effecting the outcome, i.e., allowing for {\em partial} efficacy,
where being decisive amounts to being {\em fully} efficacious.
We can formalize a player's degree of efficacy via the concept of a division's {\em loyal children}.
Call a division $\mathbb{S}$ {\em winning} if its outcome is {\sc yes} and {\em losing} if its outcome is {\sc no}.
For winning divisions, we say that a division $\mathbb{T}$ is a {\em loyal child} of $\mathbb{S}$ (and $\mathbb{S}$ is 
a {\em loyal parent} of $\mathbb{T}$) if and only if $S=T\cup \{j\}$.
That is, $\mathbb{T}$ is identical to $\mathbb{S}$ except that exactly one less player 
votes {\sc yes} in $\mathbb{T}$ than in $\mathbb{S}$.
The nomenclature {\em loyal} refers to the fact that $\mathbb{S}$ and $\mathbb{T}$ have the same outcome.
Symmetrically, for losing divisions, we say that $\mathbb{T}$ is a {\em loyal child} of $\mathbb{S}$ (and $\mathbb{S}$ 
is a {\em loyal parent} of $\mathbb{T}$)
if and only if $S=T\setminus \{j\}$.
That is, $\mathbb{T}$ is identical to $\mathbb{S}$ except that exactly one less player votes {\sc no} in $\mathbb{T}$ than in $\mathbb{S}$.
Moreover, we call a division's {\em loyal descendants} those divisions that are its loyal children, their loyal children, and so on.%
\footnote{For a defence of the notion of degrees of efficacy, and hence the conceptual foundations of RM, see Abizadeh~(2021).
For the motivation for the specific design of RM, see Abizadeh \& Vetta (2021).}

RM is then defined by specifying the division efficacy score recursively, for SVGs, as:
$$
\alpha_i^{RM}(\mathbb{S}) =
\begin{cases}
1& \mathrm{if\ } i \mathrm{\ is\ decisive\ in\ } \mathbb{S}\\
0 & \mathrm{if\ }  i \mathrm{\ is \ not\ successful \ in\ } \mathbb{S} \\
\frac{1}{|LC(\mathbb{S})|}\cdot \sum_{\mathbb{\hat{S}}\in LC(\mathbb{S})} \alpha_i(\hat{\mathbb{S}}) & \mathrm{otherwise}
\end{cases}
$$
where $LC(\mathbb{S})$ denotes a division $\mathbb{S}$'s set of loyal children in $\mathcal{D}$.

The division weight is interpreted (as with PB) as the division's probability $\mathbb{P}(\mathbb{S})$.
Since RM itself is a generalized (not specifically a priori) measure, we mark its a priori version by labelling it as RM'.
A priori voting power under RM again assumes equiprobable divisions;
the division weight is therefore equal to:

$$
\gamma^{RM'}(\mathbb{S}) = \mathbbm{p}(\mathbb{S}) = \frac{1}{|\mathcal{D}|} = \frac{1}{2^n}
$$

Whereas PB represents a probability (the player's probability of being decisive),
RM represents an expected value, namely, the player's expected efficacy
(which is a function of the player's degree of efficacy in each division weighted by the division's probability).
Note that because RM's division efficacy score tracks partial efficacy, the measure is {\em not} strategy symmetric;
hence the familiar shortcut is unavailable, and {\sc no}-efficacy must be accounted for separately from {\sc yes}-efficacy.

\section{The Subadditivity Blocker Postulates}\label{sec:sub-postulate}

In this section we study the subadditivity blocker postulates;
we consider the minimum-power blocker postulates and the added-blocker postulate
in Sections~\ref{sec:BPP} and~\ref{sec:blocker-postulates}, respectively.
Since the subadditivity blocker postulates concern blocs of voters that include blockers (vetoers),
we set up our analysis by first considering three {\em bloc} postulates not involving blockers.

\subsection{Three Bloc Postulates}

The ``conventional wisdom that the whole is greater than -- or at least equal to --
the sum of its parts'' might suggest that it would be paradoxical if the a priori voting power of a bloc of voters 
turned out to be less than the sum of the a priori voting power of each individual bloc member prior to forming the bloc (Brams 1975: 178).
We can formalize this conventional wisdom via a superadditivity postulate, concerning the {\em lower bounds} of a bloc's voting power, as follows.
Let $\hat{\mathcal{G}}$ be the voting game derived from $\mathcal{G}$ when a subset of players $I\subseteq [n]$ form a voting bloc.
We model the formation of a bloc as all members of $I$ fully donating their votes to a single lead member,
effectively rendering the donating players dummies who are then deleted from the game.
(Recall that, by one of the dummy postulates, every player's a priori voting power remains the same if a dummy is deleted.)
Consider the case of $I = \{i,j\}$, $|I|=2$.
The full donation from $j$ to $i$ induces a new (monotonic) game $\hat{\mathcal{G}}$ given by:
\begin{align*}
S\cup\{i,j\} \in \hat{\mathcal{W}} &\iff S\cup\{i,j\} \in \mathcal{W} \nonumber\\
S\cup\{i\} \in \hat{\mathcal{W}} &\iff S\cup\{i,j\} \in \mathcal{W} \nonumber\\
S\cup\{j\} \in \hat{\mathcal{W}} &\iff S \in \mathcal{W} \nonumber\\
S \in \hat{\mathcal{W}} &\iff  S \in \mathcal{W} \nonumber
\end{align*}
For the case $|I|\ge 3$ we generate $\hat{\mathcal{G}}$ simply by iterating this transformation.

A measure of voting power $\Psi$ (where $\hat{\Psi}_{i} $ is $i$'s voting power in $\hat{\mathcal{G}}$)
satisfies the {\em superadditivity bloc postulate} if, for any bloc $I\subseteq [n]$:
\begin{enumerate}
 \item[{\sc (spb)}] $\hat{\psi}_I \geq \sum_{i\in I} \psi_i$. 
\end{enumerate}
\noindent As Felsenthal and Machover (1988: 224-231) have argued, however,
such a postulate would be poorly motivated for measures of voting power,
and its violation not truly paradoxical. (Indeed, all three of our candidate measures would violate such a postulate.)
There are two basic reasons for this.
The first is that players who form a bloc lose their ability to act as separate individuals,
thus foreclosing possible strategies that otherwise might have been available to them.
It is therefore unreasonable to expect a bloc's power always to equal or exceed the sum of its members' power individually.
Call this argument for the unreasonability of a superadditivity postulate the {\em loss-of-freedom} rationale.

The second argument stems from a {\em decreasing-marginal-returns} dynamic.
Decreasing marginal returns, a discrete analogue of {\em concavity}, is a central concept in economics.
Individually, it applies when each additional unit of effort yields less incremental benefit than the previous unit.
Collectively,  it applies when the effort of an additional individual yields less incremental benefit when added to a larger group 
than a smaller group. For voting games, this effect is widespread. For example, a voter may be decisive in a
small bloc but not decisive in a larger bloc.
More generally, {\em increasing} the number of voters who vote with a voter above the minimum sufficient to ensure success
may {\em decrease} the efficacy of the voter.

Neither the loss-of-freedom nor the decreasing-marginal-returns rationale, however, rules out all expectations concerning lower bounds on a bloc's voting power.
For example, it is reasonable to expect a bloc to be just as powerful as any member would have been individually on its own:
on the one hand, the bloc as a whole has just as much freedom as any of its individual members would have had on their own;
on the other, there is no reason why adding or donating one voter's power to another would {\em diminish} the latter's individual power
(as noted, forming a bloc can be represented as all members transferring their voting power to a lead member).
We can formalize this expectation as follows.
Again, let $\hat{\mathcal{G}}$ be the voting game derived from $\mathcal{G}$ by forming a voting bloc $I\subseteq [n]$.
A measure of voting power $\Psi$ satisfies the {\em minimum-power bloc postulate} if, for any bloc $I\subseteq [n]$:
\begin{enumerate}
 \item[{\sc (mpb)}] $\hat{\psi}_I \geq \max_{i\in I} \psi_i$. 
\end{enumerate}

\noindent Felsenthal and Machover (1998: 255-56) have already shown that PB and SS satisfy the postulate,%
\footnote{Although they only show this for blocs of two players, the result follows for blocs of any size by applying their result sequentially.}
while Abizadeh and Vetta (2021) show that RM satisfies it.

The superadditivity bloc postulate, which we rejected, and the minimum-power bloc postulate, which we accept,
both concern the {\em lower} bounds of a bloc's power.
Are there reasonable expectations about {\em upper} bounds,
motivated by the loss-of-freedom and decreasing-marginal-returns rationales?
The most stringent expectation would be that a bloc's voting power
never be greater than the sum of its members' voting power prior to forming the bloc.
A measure of voting power $\Psi$ would meet this expectation if it satisfied the {\em subadditivity bloc postulate},
i.e., if, for any bloc $I\subseteq [n]$:
\begin{enumerate}
 \item[{\sc (sbb)}] $\hat{\psi}_I \leq \sum_{i\in I} \psi_i$. 
\end{enumerate}

It is not reasonable, however, to expect measures of voting power to satisfy the subadditivity bloc postulate.
The reason is because the loss-of-freedom and decreasing-marginal-returns dynamics
may in some circumstances be undercut or neutralized by two corresponding counterveiling dynamics.
First, as Felsenthal and Machover (1998: 229) have argued,
when a bloc of at least three members is formed,
whether the bloc's voting power is less or greater than the sum of its members' original, pre-bloc voting power
will often depend on whether, in the original voting structure,
the divisions in which any given individual would-be bloc member was efficacious tend to be ones in which other would-be bloc members vote against or with each other.
From amongst divisions in which would-be bloc members are efficacious,
the higher the proportion of divisions in which some members vote {\em against} each other,
the more we should expect the sum of their individual voting powers to be {\em higher} relative to the bloc's voting power
(because, once the bloc is formed, these ``high-efficacy'' divisions no longer exist);
by contrast, the higher the proportion of divisions, from amongst those in which would-be members are efficacious,
in which would-be members vote {\em together},
the more we should expect the sum of individual voting powers to be {\em less} relative to the bloc's voting power --
and hence the more we should expect the voting structure to induce a violation of the subadditivity bloc postulate.
In essence, this latter, {\em inefficacious-dissension} dynamic neutralizes the loss-of-freedom dynamic.

Second, decreasing marginal returns may sometimes (for example with small voting blocs) be counterveilled by {\em increasing marginal returns}.
The latter dynamic may arise because increasing the proportion of divisions in which more voters vote with a given voter
often increases (and, for monotonic games, never reduces) the proportion of divisions in which the voter is {\em successful}
and, hence, potentially efficacious or decisive.

Indeed, none of our candidate measures satisfies the subadditivity bloc postulate.
For example, consider three-voter unanimity rule in which the three voters subsequently form a bloc.
This is clearly a case displaying both inefficacious dissension
(each voter is decisive in only two divisions, in both of which other would-be bloc members all vote together)
and increasing marginal returns (the bloc is successful in all divisions).
Prior to forming the bloc, for each voter $PB_i = \frac{1}{4}$, such that $\sum{PB_i} = \frac{3}{4}$.
But once the bloc is formed, $PB_I = 1$, in violation of subadditivity.
For SS and RM, consider the weighted voting game $\mathcal{G}=\{2; 1,1,2,2,2\}$
where there are five voters with weights $\{1,1,2,2,2\}$ and a quota of $2$.
The reader may verify that for both SS and RM, $\hat{\psi}_{\{1,2\}} > \psi_1 +\psi_2$.
(In particular, $\hat{SS}_{\{1,2\}} = \frac{1}{4} \ > \ SS_{1} + SS_{2} = \frac{1}{20} + \frac{1}{20} = \frac{1}{10}$;
and $\hat{RM'}_{\{1,2\}} = \frac{19}{64} \ > \ RM'_{1} + RM'_{2} = \frac{41}{320} + \frac{41}{320} = \frac{41}{160}$.)


Thus, whilst the loss-of-freedom and decreasing-marginal-returns rationales render unreasonable the expectation that
the bloc's power {\em always} be greater than or equal to the sum of its members (the superadditivity bloc postulate),
they do not rule out the possibility that, on {\em some} occasions, a bloc's voting power may indeed be greater.

\subsection{Two Subadditivity Blocker Postulates}

We are nevertheless able to specify reasonable, weaker expectations about a bloc's upper bounds,
grounded in the loss-of-freedom and decreasing-marginal-returns rationales, 
under conditions in which the countervailing inefficacious-dissension and/or increasing-marginal-returns rationales are neutralized or overwhelmed.
We shall specify two such weaker postulates, concerning a bloc's upper bounds,
when the bloc contains at least one blocker (or vetoer).
We say that a player $i$ is a {\sc yes}-{\em blocker} if for every $S\in \mathcal{W}$ we have $i\in S$.
Consequently, if $i$ votes {\sc no} then the outcome is {\sc no}. That is, $i$ can veto a {\sc yes}-outcome.
Similarly, a player $j$ is a {\sc no}-{\em blocker} if for every $S\notin \mathcal{W}$ we have $j\in \bar{S}$.
That is, if $j$ votes {\sc yes} then the outcome is {\sc yes} and $j$ can block or veto a {\sc no}-outcome.

Let $\hat{\mathcal{G}}$ be the voting game derived from $\mathcal{G}$ by forming a voting bloc $I\subseteq [n]$.
Then a measure of voting power $\Psi$ satisfies the {\em strong subadditivity blocker postulate} if:
\begin{enumerate}
 \item[{\sc (sbk-1)}] $\hat{\psi}_I \leq \sum_{i\in I} \psi_i$ for any $I$ containing a {\sc yes}-blocker $b$. 
 \item[{\sc (sbk-2)}] $\hat{\psi}_I \leq \sum_{i\in I} \psi_i$ for any $I$ containing a {\sc no}-blocker $b$. 
\end{enumerate}
That is, the a priori voting power of $I$ in $\hat{\mathcal{G}}$ is no greater than
the sum of the a priori voting power of each of its members separately in $\mathcal{G}$, provided $I$ 
contains a {\sc yes}-blocker (or {\sc no}-blocker).

Why would the strong subadditivity blocker postulate be reasonable?
On the one hand, the loss-of-freedom dynamic, which favours subadditivity, is still in place:
since a bloc member $j \in I$ is forced to coordinate its vote with the bloc,
any contributions it could have made in divisions in which members would have voted against each other are excluded;
only the player's contribution when the bloc votes together counts.
But the presence of a blocker $b$ also puts into play the counterveilling inefficacious-dissension dynamic:
if the blocker is a {\sc yes}-blocker, for example,
other would-be bloc members could not have been efficacious on their own in any division in which the blocker voted {\sc no} and they voted against it.
So the loss-of-freedom dynamic on its own is insufficient to secure the subaddivity blocker postulate.
On the other hand, the presence of a blocker $b$ -- let us say a {\sc yes}-blocker --
{\em strengthens} the decreasing-marginal-returns dynamic.
When $j$ transfers its voting power to a bloc that contains a {\sc yes}-blocker,
its presence cannot help increase the bloc's division efficacy score on the {\sc no}-side any further than the score would have been without $j$:
when the bloc votes {\sc no}, the outcome is {\sc no} regardless of whether or not $j$ joins the bloc.
That is, for {\sc (sbk-1)}, the presence of player $j$ can help increase the bloc's voting power
only on the basis of divisions in which the bloc votes together {\bf and} when the {\sc yes}-blocker $b$ votes {\sc yes}.
Thus $j$'s marginal contributions will be even less.
A symmetric argument motivates {\sc (sbk-2)}.

We can further weaken the subadditivity blocker postulate by considering only the case in which {\em every} member of the bloc is a blocker.
A measure of voting power $\Psi$ satisfies the {\em weak subadditivity blocker postulate} if:
\begin{enumerate}
 \item[{\sc (wbk-1)}] $\hat{\psi}_I \leq \sum_{i\in I} \psi_i$ for any $I$ containing only {\sc yes}-blockers. 
 \item[{\sc (wbk-2)}] $\hat{\psi}_I \leq \sum_{i\in I} \psi_i$ for any $I$ containing only {\sc no}-blockers. 
\end{enumerate}
That is, the a priori voting power of $I$ in $\hat{\mathcal{G}}$ is no greater than
the sum of the a priori voting power of each of its members separately in $\mathcal{G}$, 
provided every member of $I$ is a {\sc yes}-blocker (or every member of $I$ is a {\sc no}-blocker).
Observe that {\sc (wbk-1)} is indeed weaker than  {\sc (sbk-1)} because it only need hold in the restricted case 
where the bloc $I$ consists entirely of {\sc yes}-blockers.
For example, the decreasing-marginal-returns rationale may hold only when the added agent 
$j$ is itself a {\sc yes}-blocker (as well as the other members of the bloc) and need not hold when
$j$ is not a blocker.

\subsection{Which Measures Satisfy the Subadditivity Blocker Postulates?}
We next determine whether or not PB, SS and RM satisfy the subadditivity blocker postulates.
(Proofs for all technical theorems in the paper are deferred to the appendix.)

\begin{theorem}\label{thm:PB-sub}
PB fails to satisfy the weak subadditivity blocker postulate (and, thus, fails to satisfy the strong subadditivity blocker postulate).
\end{theorem}

\begin{theorem}\label{thm:SS-sub}
SS satisfies the weak subadditivity blocker postulate but does not satisfy the strong subadditivity blocker postulate.
\end{theorem}

\begin{theorem}\label{thm:RM-sub}
RM satisfies the strong subadditivity blocker postulate (and, thus, satisfies the weak subadditivity blocker postulate).
\end{theorem}

\section{The Minimum-Power Blocker Postulates}\label{sec:BPP}
Felsenthal and Machover (1998: 264) introduce what they call the {\em blocker's share postulate},
which is satisfied by a measure $\Psi$ if the share of any {\sc yes}-blocker's a priori voting power, out of the sum total of all players' a priori voting power,
is at least as great as the reciprocal of the number of players in any {\sc yes}-successful set of voters (and similarly for {\sc no}-blockers):
\begin{enumerate}
\vspace{-.25cm}
 \item[{\sc (bsp-1)}] If $b\in S$ is a {\sc yes}-blocker, then  $\frac{\psi_{b}}{\sum_{i\in [n]} \psi_i} \geq  \frac{1}{|S|}$, for any $S \in \mathcal{W}$.
 \item[{\sc (bsp-2)}] If $b\in \bar{T}$ is a {\sc no}-blocker, then $\frac{\psi_{b}}{\sum_{i\in [n]} \psi_i} \geq  \frac{1}{|\bar{T}|}$, for any $T \notin \mathcal{W}$.
\end{enumerate}
These lower bounds of a blocker's voting power are of course most stringent when $S$ is the smallest possible set of successful {\sc yes}-voters $S^*$,
and $\bar{T}$ is the smallest possible set of successful {\sc no}-voters $\bar{T}^*$.
In effect, the postulate requires that for any {\sc yes}-blocker, its share of a priori voting power, out of the sum of all players' a priori voting power,
be no less than $\frac{1}{|S^*|}$.
Felsenthal and Machover then prove that PB violates the postulate,
i.e., that when PB is normalized such that all players' scores sum to 1,
which yields each player's relative voting power according to what they call the Banzhaf {\em index},
a {\sc yes}-blocker's relative a priori voting power according to the Banzhaf index may be less than $\frac{1}{|S^*|}$,
whereas SS (which is itself already a relative index, since $\sum_{i\in [n]} {\SS}_i=1$) satisfies it.
It can be shown that RM, like PB, also violates the blocker's share postulate.

Does this speak against PB and RM? It does not:
the postulate is unmotivated for voting power.
The intuition behind the postulate does not concern voting power as such but, rather, its expected {\em value}
(as is implicit to Felsenthal and Machover's justification for the postulate).%
\footnote{Felsenthal and Machover put it in terms of their distinction between I-power and P-power.}
Consider a non-voter for whom the value of a {\sc yes}-outcome is equal to $f>0$,
and who would therefore be willing to spend up to $f$ to buy the votes of a set of voters capable of ensuring a {\sc yes}-outcome.
This set must include any {\sc yes}-blocker, if there is one.
Assume the voting structure has at least one {\sc yes}-blocker,
the non-voter knows the voting structure,
but has no information about the distribution of player preferences (modelling the a priori case).
Without information about voter inclinations,
the most efficient strategy is to bribe the smallest possible set of successful {\sc yes}-voters $S^*$,
to minimize the total bribe necessary to secure the desired outcome.
What is the (subjective) expected value, to the non-voter, of the {\sc yes}-blocker's vote,
i.e., the value of bribing the {\sc yes}-blocker rather other players to realize the desired {\sc yes}-outcome?
If every member of $S^*$ is a {\sc yes}-blocker,
then the expected value of any {\sc yes}-blocker's vote to the non-voter will be equal to that of any other {\sc yes}-blocker,
which implies the non-voter would be willing to offer each {\sc yes}-blocker a bribe of up to $\frac{1}{|S^*|}\cdot f$.
If, by contrast, not all members of $S^*$ are {\sc yes}-blockers,
then the smallest possible set of successful {\sc yes}-voters will not be unique,
i.e., there will more than one such minimal set of voters.
And since only a {\sc yes}-blocker will be a member of every such minimal set,
the expected value of its vote to the non-voter will at least as great as that of any other potential member of these minimal sets.
Therefore, $\frac{1}{|S^*|}\cdot f$ is the minimum value of a {\sc yes}-blocker's vote to the non-voter,
and $\frac{1}{|S^*|}$ its minimum {\em relative} value.
And this is precisely what the blocker's share postulate says:
that the relative {\em value} of a {\sc yes}-blocker's a priori voting power should be at least $\frac{1}{|S^*|}$.
We therefore conclude that the blocker's share postulate is appropriate not for measures of a priori {\em voting power},
but, rather, for measures of the expected {\em value} of a player's a priori voting power.
(The fact that SS satisfies the blocker's share postulate provides support,
in other words, for the view that, considered as an a priori index, it is best interpreted,
not as an index of relative a priori voting power,
but, rather, as an index of a player's expected payoff assuming a cooperative game with transferable utility).%
\footnote{On SS as a bribe index, see Morriss (2002);
on the equivalent expected payoff interpretation, see Felsenthal and Machover (1998).}
 
It is possible, however, to reformulate the postulate in a way that would be appropriate for measures of voting power,
by focussing on a player's a priori power itself, rather than its share of overall power.
The key is to compare its a priori power against the voting power
that a measure of voting power would assign to a dictator in a dictator-rule SVG.
To formalize this, let $S$ be a {\sc yes}-successful set in $\mathcal{G}$ ($S \in \mathcal{W}$),
$\bar{T} = [n] \setminus T$ be a {\sc no}-successful set in $\mathcal{G}$ ($T \notin \mathcal{W}$),
and $\psi_d$ be the voting power of player $d$ in $\mathcal{G'}$,
where $\mathcal{G'}$ is any dictator-rule SVG and $d$ is the dictator.

A measure of voting power $\Psi$ satisfies the {\em strong minimum-power blocker postulate} if:
\begin{enumerate}
 \item[{\sc (smp-1)}] If $b\in S$ is a {\sc yes}-blocker, then $\psi_{b} \geq  \frac{\psi_d}{|S|}$, for any $S \in \mathcal{W}$.
 \item[{\sc (smp-2)}] If $b\in \bar{T}$ is a {\sc no}-blocker, then $\psi_{b} \geq  \frac{\psi_d}{|\bar{T}|}$, for any $T \notin \mathcal{W}$.
\end{enumerate}

A measure of voting power $\Psi$ satisfies the {\em weak minimum-power blocker postulate} if:
\begin{enumerate}
 \item[{\sc (wmp-1)}] If every voter in $S$ is a {\sc yes}-blocker, then $\psi_{b} \geq  \frac{\psi_d}{|S|}$, for all $b\in S$, for any $S \in \mathcal{W}$.
 \item[{\sc (wmp-2)}] If every voter in $\bar{T}$ is a {\sc no}-blocker, then $\psi_{b} \geq  \frac{\psi_d}{|\bar{T}|}$, for all $b\in \bar{T}$, for any $T \notin \mathcal{W}$.
\end{enumerate}

The strong minimum-power blocker postulate obviously implies the weak minimum-power blocker postulate.
We remark that $\psi_d=1$ for a dictator $d$ for each of the three voting measures studied in this paper.
In fact, multiplying $\psi_i(\mathbb{S})$ by a fixed scalar for each voter $i$ and each subset $S$
has no effect on the structure of the voting game nor on the satisfaction of any of the postulates.
The fact that $\psi_d=1$ for each of three voting measures essentially means their scaling factors are all identical.

To understand the minimum-power blocker postulates, let's begin with a sanity check.
Suppose $S=\{b\}$ with cardinality $1$. If $\{b\}$ is a {\sc yes}-successful set {\bf and} $b$ is a {\sc yes}-blocker,
then a division $\mathbb{R}=(R, \bar{R})$ is {\sc yes}-successful if and only if $b\in R$.
That is, $b$ is a dictator! It follows that $\psi_{b} = \psi_d \ge \frac{\psi_d}{1}$ and the weak minimum-power blocker postulate {\sc (wmp-1)}
holds for $S=\{b\}$. A similar argument applies for {\sc (wmp-2)}.

Suppose, by contrast, that $S$ is a {\sc yes}-successful set containing more than one player, but each is a {\sc yes}-blocker.
Then $\mathbb{R}=(R, \bar{R})$ is {\sc yes}-successful if and only if $S\subseteq R$.
Thus, collectively $S$ is a dictatorship. 
But, in addition, each voter $b\in S$ has the individual power to veto any {\sc yes}-outcome.
Thus the agents in $S$ have the option of choosing to act collectively as dictator,
and individually each must be in any {\sc yes}-successful set.
It is thus reasonable to expect the voting power of any member of $S$ does 
satisfy $\psi_{b} \geq  \frac{\psi_d}{|S|}$ and the weak minimum-power blocker postulate holds.

A similar argument applies for the strong minimum-power blocker postulate.
The agents in $S$ can choose to act collectively as a dictator. But since $S$ may contain non-blockers, this
is only because a {\sc yes}-blocker $b$ has the power to veto any {\sc yes}-outcome. Thus $b$ should have power at least as large as any other 
member of $S$ and, specifically, at least as the large as the average, and hence $\psi_{b} \geq  \frac{\psi_d}{|S|}$.

\subsection{Which Measures Satisfy the Minimum-Power Blocker Postulates?}
We next determine whether PB, SS and RM satisfy the minimum-power blocker postulates.

\begin{theorem}\label{thm:PB-smp}
PB does not satisfy the weak minimum-power blocker postulate (and, thus, does not satisfy the strong minimum-power blocker postulate).
\end{theorem}

\begin{theorem}\label{thm:SS-smp}
SS satisfies the strong minimum-power blocker postulate (and, thus, satisfies the weak minimum-power blocker postulate).
\end{theorem}

\begin{theorem}\label{thm:RM-smp}
RM satisfies the strong minimum-power blocker postulate (and, thus, satisfies the weak minimum-power blocker postulate).
\end{theorem}

\section{The Added-Blocker Postulate}\label{sec:blocker-postulates}
We conclude with the added-blocker postulate, which concerns changes to other players' a priori voting power when a blocker is added to a game.

The first step is to formulate an added-blocker postulate that is appropriate for a priori voting power in general.
Given a game $\mathcal{G}=([n], \mathcal{W})$, let $\mathcal{G}^Y=([n]\cup \{0\}, \mathcal{W}^Y)$ be the game resulting 
from adding an {\em added {\sc yes}-blocker}, i.e., 
a new player $0$ that is a {\sc yes}-blocker but who otherwise does not affect the original voting structure. 
Specifically, $\mathcal{W}^{Y}=\{S\cup\{0\}: \forall S\in  \mathcal{W}\}$.
Similarly, let $\mathcal{G}^N=([n]\cup \{0\}, \mathcal{W}^N)$ be the game resulting from adding an {\em added {\sc no}-blocker}~$0$.
Specifically, $\mathcal{W}^{N}=\{S\cup\{0\}: \forall S\}\cup \{S: \forall S\in \mathcal{W}\}$.

Felsenthal and Machover (1988: 266-275) argue that any reasonable measure of voting power $\Psi$ must satisfy the 
following postulate for a priori voting power. For any pair of players $i$ and $j$,

\begin{enumerate}
\vspace{-.25cm}
 \item[{\sc (add-0)}] ~$\frac{\psi_i(\mathcal{G})}{\psi_j(\mathcal{G})} = \frac{\psi_i(\mathcal{G}^Y)}{\psi_j(\mathcal{G}^Y)}$
\end{enumerate}
That is, the relative measures of a priori voting power for $i$ and $j$ should be unaffected by an added {\sc yes}-blocker.
They argue ``there is nothing at all to imply that the addition of the new'' {\sc yes}-blocker ``is of greater relative advantage 
to some of the voters'' of the original game than to others,
because there is ``no reasonable mechanism that would create a {\em differential} effect'' (Felsenthal and Machover 1998: 267).
They then show that PB satisfies this postulate, but SS violates it, and, on this basis,
conclude that the latter cannot be considered a valid index of a priori voting power.%
\footnote{Where voting power is understood as the capacity to influence voting outcomes (Felsenthal and Machover 1998: 267-275).}

But there is a problem: their specification is asymmetric between {\sc yes}-voting and {\sc no}-voting power.\footnote{A similar issue 
arises with the {\em bicameral postulate} (Felsenthal et al. 1998).}
Contrary to their assertion,
in general we do have good reasons to expect an added {\sc yes}-blocker sometimes to have differential relative impact on players' a priori 
voting power as a whole,
depending on the relative importance, to each player's total voting power, of its {\sc yes}- as opposed to {\sc no}-voting power.
This is because an added {\sc yes}-blocker may diminish the relative significance or share of {\sc yes}-voting power within a player's 
total voting power.

We should expect this potential asymmetry between {\sc yes}- and {\sc no}-voting power to be neutralized only for measures of voting power that, 
like PB and SS, give a positive efficacy score only in cases of (full) decisiveness.
Such measures, by ignoring partial efficacy, effectively render a player's a priori {\sc yes}- and {\sc no}-voting power perfectly symmetrical, 
that is, $\psi_i^+ = \psi_i^-$:
any player that is {\sc yes}-decisive in a winning division will also be {\sc no}-decisive in the corresponding losing division
in which the only difference is that player's vote.
By contrast, this symmetry between {\sc yes}- and {\sc no}-voting power will not hold for measures that, like RM, take 
degrees of efficacy into account.
A player that is only partially efficacious in a winning division will not be efficacious at all in the corresponding division in which all 
other players' votes are held constant,
because, not being (fully) decisive, the player's switch from {\sc yes} to {\sc no} will not change the outcome --
which switches the player from successful in one division to unsuccessful in the other.
RM is not, in other words, strategy symmetric.
The implication is that, if a player's total voting power relies more heavily on partial efficacy in winning divisions than does that of another player,
then an added {\sc yes}-blocker may have a disproportionately negative impact on the former than on the latter.
We therefore have no reason in general to expect an added {\sc yes}-blocker never to result in some players' relative advantage.

By contrast,
we have every reason to expect that an added {\sc yes}-blocker will be of no relative advantage to players' a priori {\em {\sc yes}-voting power}
and that an added {\sc no}-blocker will be of no relative advantage to players' a priori {\em {\sc no}-voting power} in particular.
This induces an easy fix to Felsenthal and Machover's proposal,
so as to yield a postulate appropriate for all efficacy measures in general (and not just decisiveness measures).
We simply reformulate their postulate to distinguish between {\sc yes}-voting power and {\sc no}-voting power.
Accordingly, we say that a measure of voting power $\Psi$ satisfies the {\em added-blocker postulate} if, for any pair of players $i$ and~$j$, the following 
conditions hold for a priori voting power:

\begin{enumerate}
\vspace{-.25cm}
 \item[{\sc (add-1)}] ~$\frac{\psi^+_i(\mathcal{G})}{\psi^+_j(\mathcal{G})} = \frac{\psi^+_i(\mathcal{G}^Y)}{\psi^+_j(\mathcal{G}^Y)}$, and
 \item[{\sc (add-2)}] ~$\frac{\psi^-_i(\mathcal{G})}{\psi^-_j(\mathcal{G})} = \frac{\psi^-_i(\mathcal{G}^N)}{\psi^-_j(\mathcal{G}^N)}$
\end{enumerate}

{\sc (add-1)} says that the relative measures of a priori {\sc yes}-voting power are unaffected by an added {\sc yes}-blocker.
Hence, we call this element of the postulate the {\em added-{\sc yes}-blocker postulate} for any pair of players $i$ and $j$.
{\sc (add-2)} says the relative measures of a priori {\sc no}-voting power are unaffected by an added {\sc no}-blocker.
We similarly call this element the {\em added-{\sc no}-blocker postulate}.

\subsection{Which Measures Satisfy the Added-Blocker Postulate?}
To finish we determine whether or not PB, SS and RM satisfy the added-blocker postulate.

\begin{theorem}\label{thm:added-PB}
PB satisfies the added-blocker postulate.
\end{theorem}

\begin{theorem}\label{thm:added-SS}
SS does not satisfy the added-blocker postulate.
\end{theorem}

\begin{theorem}\label{thm:added-RM}
RM satisfies the added-blocker postulate.
\end{theorem}

\section{Conclusion}\label{sec:conc}

We have specified and motivated five reasonable postulates about a priori voting power that a measure of voting power should satisfy:
the strong and weak subadditivity blocker postulates, the strong and weak minimum-power blocker postulates,
and the added-blocker postulate.
We further showed that the classic Penrose-Banzhaf measure violates the two subadditivity blocker postulates and the two minimum-power blocker postulates,
while the classic Shapley-Shubik index violates the strong subadditivity blocker postulate and the added-blocker postulate.
These violations weaken the plausibility of PB and SS as measures of voting power.
By contrast, the Recursive Measure, alone amongst the three measures studied here,
withstands full scrutiny: it satisfies all five postulates.
We take this finding considerably to buttress the plausibility of RM as a measure of voting power.

\section{Appendix of Proofs}\label{sec:appendix}

\subsection*{Proofs for Section~\ref{sec:sub-postulate}}

\restatetheorem{thm:PB-sub}
PB fails to satisfy the weak subadditivity blocker postulate (and, thus, fails to satisfy the strong subadditivity blocker postulate).
\end{theorem}
\begin{proof}
Let $\mathcal{G}$ be a three-player unanimity-rule voting game,
and let $\hat{\mathcal{G}}$ be the game derived from $\mathcal{G}$ when all three players form a unanimous bloc $I$.
By the unanimity-rule, all three players are {\sc yes}-blockers in $\mathcal{G}$.
Now PB gives a value to each player of $\frac{1}{4}$ and thus $\sum_{i=1}^3 PB_i = \frac{3}{4}$.
But $I$ is a dictator in $\hat{\mathcal{G}}$, so $\hat{PB}_I = 1 > \frac{3}{4}$, in violation of each subadditivity blocker postulate.
\end{proof}

\restatetheorem{thm:SS-sub}
SS satisfies the weak subadditivity blocker postulate but does not satisfy the strong subadditivity blocker postulate.
\end{theorem}
\begin{proof}
First we show SS does not satisfy the strong subadditivity blocker postulate via a counterexample.
Consider the weighted voting game $\mathcal{G}$=$\{3: 2,1,1\}$. 
Observe that voter $1$ is a {\sc yes}-blocker since voters $2$ and $3$ have voting weight $1+1$ which is smaller than the quota of $3$.
Here $SS_1 = \frac{2}{3}$ and $SS_2=SS_3=\frac{1}{6}$.
Now let $\hat{\mathcal{G}}$ be the game derived from $\mathcal{G}$ in which the first two players form a bloc $I=\{1,2\}$.
It follows that $I$ is a dictator in $\hat{\mathcal{G}}$, so $\hat{SS}_I = 1 
> \frac{2}{3}+\frac{1}{6}=\frac{5}{6}$, in violation of the postulate.

Second, we prove SS does satisfy the weak subadditivity blocker postulate.
It suffices to show {\sc (wbk-1)} holds.
Assume $I=\{i,j\}$. Let {\sc yes}-blocker $j$ donate to {\sc yes}-blocker $i$. Then $j$ becomes a dummy.
Take an ordering $\sigma$ of the agents. We have three cases.\\
(i) Let $j$ be the pivotal voter in the ordering $\sigma$ for $\mathcal{G}$.
Let $S$ be the set of agents before $j$ in the 
ordering. Since $j$ is decisive in $S\cup \{j\}$ in $\mathcal{G}$  but is a dummy in $\hat{\mathcal{G}}$, 
it must be the case that $S$ is a {\sc yes}-successful set in $\hat{\mathcal{G}}$.
In particular, $i$ must appear before $j$ in $\sigma$ since $i$ is a {\sc yes}-blocker.
Thus $i\in S$.
Since $S$ is {\sc yes}-successful in $\hat{\mathcal{G}}$ there must be some agent in $S$ that is now decisive.
This may or may not be agent~$i$.\\
(ii) Let $i$ be the pivotal voter in the ordering $\sigma$ for $\mathcal{G}$.
Let $S$ be the set of agents before $i$ in the ordering.
So $S$ is {\sc yes}-unsuccessful and $S\cup \{i\}$ is {\sc yes}-successful in $\mathcal{G}$.
This must still be the case after the donation from $j$ to $i$. Thus $i$ remains the pivotal voter in the ordering $\sigma$
for $\hat{\mathcal{G}}$.\\
(iii) Let $\ell\neq i,j$ be the pivotal voter in the ordering $\sigma$ for $\mathcal{G}$.
Let $S$ be the set of agents before $\ell$ in the ordering.
Since $S\cup \{\ell\}$ is {\sc yes}-successful in $\mathcal{G}$,
it must be the case that $i$ and $j$ are in $S$ since they are {\em both} {\sc yes}-blockers.
(Note, this is where the distinction between the weak and strong postulates is important.)
But then, by definition, $\mathbb{S}$ and $\mathbb{S}\cup \mathbbm{\{\ell\}}$ have the same outcomes in $\mathcal{G}$ and $\hat{\mathcal{G}}$.
Thus $\ell$ remains pivotal in the ordering $\sigma$ for $\hat{\mathcal{G}}$.\\
It immediately follows that $\hat{\psi}_I \leq  \psi_i+\psi_j$. 
Iterating this argument, we have that $\hat{\psi}_I \leq \sum_{i\in I} \psi_i$ for
any $I$ consisting only of three or more blockers.
\end{proof}

\restatetheorem{thm:RM-sub}
RM satisfies the strong subadditivity blocker postulate (and, thus, satisfies the weak subadditivity blocker postulate).
\end{theorem}

The crux to proving Theorem~\ref{thm:RM-sub} is the following lemma.
\begin{lemma}\label{lem:crux}
Let $j$ fully donate to $i$.
If $j$ is a {\sc yes}-blocker (or a  {\sc no}-blocker), then $\hat{RM'}_i\le RM'_i+RM'_j$.
\end{lemma}
\begin{proof}
Observe that
\begin{align*}
RM'_i &\ =\ \frac{1}{2^n}\cdot \sum_{\mathbb{S}\in \mathcal{D}: i,j\notin S} \left( \alpha_i(\mathbb{S})+ \alpha_i(\mathbb{S}\cup \mathbbm{j})+ \alpha_i(\mathbb{S}\cup \mathbbm{i}) + \alpha_i(\mathbb{S}\cup \mathbbm{\{i,j\}}) \right)\\
&\ =\ \frac{1}{2^n}\cdot \sum_{\mathbb{S}\in \mathcal{D}: i,j\notin S} \left( \alpha^-_i(\mathbb{S})+ \alpha^-_i(\mathbb{S}\cup \mathbbm{j})+ \alpha^+_i(\mathbb{S}\cup \mathbbm{i}) + \alpha^+_i(\mathbb{S}\cup \mathbbm{\{i,j\}}) \right)
\end{align*}
Thus to prove $\hat{RM'}_i\le RM'_i+RM'_j$ it suffices to show that
\begin{eqnarray*}
\lefteqn{ \hat{\alpha}^-_i(\mathbb{S})+ \hat{\alpha}^-_i(\mathbb{S}\cup \mathbbm{j})+ \hat{\alpha}^+_i(\mathbb{S}\cup \mathbbm{i}) + \hat{\alpha}^+_i(\mathbb{S}\cup \mathbbm{\{i,j\}}) }\\
&\le&
\left( \alpha^-_i(\mathbb{S})+  \alpha^-_j(\mathbb{S})  \right)+
\left(  \alpha^-_i(\mathbb{S}\cup \mathbbm{j}) + \alpha^+_j(\mathbb{S}\cup \mathbbm{j}) \right) +
\left( \alpha^+_i(\mathbb{S}\cup \mathbbm{i}) + \alpha^-_j(\mathbb{S}\cup \mathbbm{i}) \right) +
\left( \alpha^+_i(\mathbb{S}\cup \mathbbm{\{i,j\}}) + \alpha^+_j(\mathbb{S}\cup \mathbbm{\{i,j\}}) \right)
\end{eqnarray*}
Take any $S$ containing neither $i$ nor $j$. 
Since $j$ is a {\sc yes}-blocker,  $\mathbb{S}$ and $\mathbb{S}\cup \mathbbm{i}$ lose in the original game.
Furthermore, by monotonicity, in the original game either:\\
(i) $\mathbb{S}\cup \mathbbm{j}$ and $\mathbb{S}\cup \mathbbm{\{i,j\}}$ both win, or\\
(ii) $\mathbb{S}\cup \mathbbm{j}$ loses and $\mathbb{S}\cup \mathbbm{\{i,j\}}$ wins, or\\
(iii) $\mathbb{S}\cup \mathbbm{j}$ and $\mathbb{S}\cup \mathbbm{\{i,j\}}$ both lose.\\
Thus there are three cases.
In the modified game these three cases imply:\\
(i) $\mathbb{S}\cup \mathbbm{i}$ and $\mathbb{S}\cup \mathbbm{\{i,j\}}$ both win but $\mathbb{S}$ and $\mathbb{S}\cup \mathbbm{j}$ both lose.\\
(ii) $\mathbb{S}\cup \mathbbm{i}$ and $\mathbb{S}\cup \mathbbm{\{i,j\}}$ both win but $\mathbb{S}$ and $\mathbb{S}\cup \mathbbm{j}$ both lose.\\
(iii) $\mathbb{S}$, $\mathbb{S}\cup \mathbbm{j}$, $\mathbb{S}\cup \mathbbm{i}$ and $\mathbb{S}\cup \mathbbm{\{i,j\}}$ all lose.\\
The first two cases are easier to deal with.
In Case (i) observe that $j$ is {\sc yes}-decisive at both $\mathbb{S}\cup \mathbbm{j}$ and $\mathbb{S}\cup \mathbbm{\{i,j\}}$
and {\sc no}-decisive at both $\mathbb{S}$ and $\mathbb{S}\cup \mathbbm{i}$. Thus
\begin{align*}
 \left(\alpha^-_i(\mathbb{S})+  \alpha^-_j(\mathbb{S})  \right) &+
\left(  \alpha^-_i(\mathbb{S}\cup \mathbbm{j}) + \alpha^+_j(\mathbb{S}\cup \mathbbm{j}) \right) +
\left( \alpha^+_i(\mathbb{S}\cup \mathbbm{i}) + \alpha^-_j(\mathbb{S}\cup \mathbbm{i}) \right) +
\left( \alpha^+_i(\mathbb{S}\cup \mathbbm{\{i,j\}}) + \alpha^+_j(\mathbb{S}\cup \mathbbm{\{i,j\}}) \right)  \\
&\ge \left( 0+ 1 \right)+ \left( 0+ 1 \right) + \left(0+ 1 \right) + \left(0+ 1) \right)\\
&\ge 4\\
&\ge
\hat{\alpha}^-_i(\mathbb{S})+ \hat{\alpha}^-_i(\mathbb{S}\cup \mathbbm{j})+ \hat{\alpha}^+_i(\mathbb{S}\cup \mathbbm{i}) 
+ \hat{\alpha}^+_i(\mathbb{S}\cup \mathbbm{\{i,j\}}) \end{align*}

Similarly, in Case (ii) observe that both $i$ and $j$ are {\sc yes}-decisive at $\mathbb{S}\cup \mathbbm{\{i,j\}}$.
Furthermore, $i$ is {\sc no}-decisive at $\mathbb{S}\cup \mathbbm{j}$ and 
$j$ is {\sc no}-decisive at $\mathbb{S}\cup \mathbbm{i}$. Thus
\begin{align*}
 \left( \alpha^-_i(\mathbb{S})+  \alpha^-_j(\mathbb{S})  \right) &+
\left(  \alpha^-_i(\mathbb{S}\cup \mathbbm{j}) + \alpha^+_j(\mathbb{S}\cup \mathbbm{j}) \right) +
\left( \alpha^+_i(\mathbb{S}\cup \mathbbm{i}) + \alpha^-_j(\mathbb{S}\cup \mathbbm{i}) \right) +
\left( \alpha^+_i(\mathbb{S}\cup \mathbbm{\{i,j\}}) + \alpha^+_j(\mathbb{S}\cup \mathbbm{\{i,j\}}) \right)  \\
&\ge \left( 0+ 0 \right)+ \left( 1+ 0 \right) + \left(0+ 1 \right) + \left(1+ 1) \right)\\
&\ge 4\\
&\ge
\hat{\alpha}^-_i(\mathbb{S})+ \hat{\alpha}^-_i(\mathbb{S}\cup \mathbbm{j})+ \hat{\alpha}^+_i(\mathbb{S}\cup \mathbbm{i}) 
+ \hat{\alpha}^+_i(\mathbb{S}\cup \mathbbm{\{i,j\}}) 
\end{align*}
Case (iii) is more complex.
Recall that in this case 
$\mathbb{S}$, $\mathbb{S}\cup \mathbbm{j}$, $\mathbb{S}\cup \mathbbm{i}$ and $\mathbb{S}\cup \mathbbm{\{i,j\}}$ all lose
in both the original game and the modified game. This implies that the {\sc yes}-efficacy score of each voter is zero
at these divisions. Thus it suffices to prove that 
\begin{eqnarray}
\hat{ \alpha}^-_i(\mathbb{S})+ \hat{\alpha}^-_i(\mathbb{S}\cup \mathbbm{j}) 
&\le&
\alpha^-_i(\mathbb{S})+  \alpha^-_j(\mathbb{S})  + \alpha^-_i(\mathbb{S}\cup \mathbbm{j}) + \alpha^-_j(\mathbb{S}\cup \mathbbm{i}) 
\end{eqnarray}
In fact, we will prove something stronger. The following two inequalities hold.
\begin{eqnarray}
\hat{ \alpha}^-_i(\mathbb{S})
&\le&
\alpha^-_i(\mathbb{S})+  \alpha^-_j(\mathbb{S})  \label{eq:suff-1}\\
\hat{\alpha}^-_i(\mathbb{S}\cup \mathbbm{j}) 
&\le& \alpha^-_i(\mathbb{S}\cup \mathbbm{j}) + \alpha^-_j(\mathbb{S}\cup \mathbbm{i}) \label{eq:suff-2}
\end{eqnarray}

To show this we need the following important fact from Abizadeh and Vetta (2021).
A {\sc no}-efficacy score $\alpha^-_i(\mathbb{S})$ for RM can be calculated by considering paths 
in the division lattice from $\mathbb{S}$ to $[\mathbbm{n}]$. 
The division lattice contains a node for each division $\mathbb{S}$. There is an arc in the lattice from
$\mathbb{S}$ to $\mathbb{S}\cup \mathbbm{j}$, for each $j\in [n]\setminus S$.
Then the {\sc no}-efficacy score $\alpha^-_i(\mathbb{S})$
is the fraction of paths from $\mathbb{S}$ to $[\mathbbm{n}]$ that contains a division at which 
$i$ is {\sc no}-decisive.

Using this fact we proceed to prove (\ref{eq:suff-1}). 
\begin{claim}\label{cl:1}
$\hat{ \alpha}^-_i(\mathbb{S}) \le \alpha^-_i(\mathbb{S})+  \alpha^-_j(\mathbb{S})$
\end{claim}
\begin{proof}
Take any path $P$ from $\mathbb{S}$ to $[\mathbbm{n}]$
that contains a division at which $i$ is {\sc no}-decisive {\em in the modified game} $\hat{\mathcal{G}}$.
Thus $P$ contributes to $\hat{ \alpha}^-_i(\mathbb{S})$. We wish to find a matching contribution to 
$\alpha^-_i(\mathbb{S})+  \alpha^-_j(\mathbb{S})$.
Again we break the analysis into cases.
\begin{enumerate}
\item $i$ is {\sc no}-decisive at $\mathbb{T}$ on path $P$ in $\hat{\mathcal{G}}$\\
By monotonicity, $i$ is {\sc no}-decisive at the highest {\sc no}-division on $P$ in $\hat{\mathcal{G}}$.
Thus, we may assume $\mathbb{T}\cup \mathbbm{k}$ on path $P$ is a {\sc yes}-division in $\hat{\mathcal{G}}$.
But $i$ is a {\sc yes}-blocker in $\hat{\mathcal{G}}$ because $j$ is a {\sc yes}-blocker in ${\mathcal{G}}$.
Hence, it must be that case that $k=i$. 

Now $\mathbb{T}\cup \mathbbm{i}$ is losing in ${\mathcal{G}}$ and $j$ is a {\sc yes}-blocker in ${\mathcal{G}}$.
But $\mathbb{T}\cup \mathbbm{i}$ is winning in $\hat{\mathcal{G}}$.
So, by definition, $\mathbb{T}\cup \mathbbm{\{i,j\}}$ is winning in ${\mathcal{G}}$.
This implies $j$ is {\sc no}-decisive at $\mathbb{T}\cup \mathbbm{i}$ on path $P$ in ${\mathcal{G}}$.
Consequently, $P$ contributes to $\alpha^-_j(\mathbb{S})$.
\item $i$ is {\sc no}-decisive at $\mathbb{T}\cup \mathbbm{j}$ on path $P$ in $\hat{\mathcal{G}}$
\begin{enumerate}
\item $\mathbb{T}\cup \mathbbm{j}$ is losing in ${\mathcal{G}}$: note that $\mathbb{T}\cup \mathbbm{\{i,j\}}$ is winning in ${\mathcal{G}}$ as it is winning in $\hat{\mathcal{G}}$.
Therefore, $i$ is also {\sc no}-decisive at $\mathbb{T}\cup \mathbbm{j}$ on path $P$ in ${\mathcal{G}}$.
Thus $P$ contributes to $\alpha^-_i(\mathbb{S})$.
\item $\mathbb{T}\cup \mathbbm{j}$ is winning in ${\mathcal{G}}$:
Now consider the {\em mirror path} $P^M$ from $\mathbb{S}$ to $[\mathbbm{n}]$ which is identical
to $P$ except the roles or $i$ and $j$ are switched (that is $i$ and $j$ swap their positions in $P$).
In particular, $P^M$ passes through the division $\mathbb{T}\cup \mathbbm{i}$. But as $j$ is a {\sc yes}-blocker in ${\mathcal{G}}$ it 
must be the case that $\mathbb{T}\cup \mathbbm{i}$ is losing in ${\mathcal{G}}$. 
As $\mathbb{T}\cup \mathbbm{\{i,j\}}$ is winning in both ${\mathcal{G}}$ and $\hat{\mathcal{G}}$,
it follows that $j$ is 
{\sc no}-decisive at $\mathbb{T}\cup \mathbbm{i}$ on the path $P^M$ in ${\mathcal{G}}$.
Consequently, $P^M$ contributes to $\alpha^-_j(\mathbb{S})$.
On the other hand, suppose $P^M$ also contributes to $\hat{\alpha}^-_i(\mathbb{S})$.
This can only happen if $i$ is
{\sc no}-decisive at $\mathbb{R}$ in $\hat{\mathcal{G}}$, where $\mathbb{R}$ is the division where $P$ and $P^M$ diverge.
By definition, this implies $\mathbb{R} \cup \mathbbm{\{i,j\}}$ is winning in ${\mathcal{G}}$.
In particular, $i$ is {\sc no}-decisive at $\mathbb{R}\cup \mathbbm{j}$ on $P$ 
in $\hat{\mathcal{G}}$. So in this case $P$ also contributes to $\alpha^-_i(\mathbb{S})$.
\end{enumerate}
Ergo, the combined contribution of $P$ and $P^M$ to $\alpha^-_i(\mathbb{S})$ is at most their combined
contribution to $\alpha^-_i(\mathbb{S})+\alpha^-_j(\mathbb{S})$. The claim follows.
\qedhere
\end{enumerate}
\end{proof}

Now we prove (\ref{eq:suff-2}). 
\begin{claim}\label{cl:2}
$\hat{\alpha}^-_i(\mathbb{S}\cup \mathbbm{j}) 
\le \alpha^-_i(\mathbb{S}\cup \mathbbm{j}) + \alpha^-_j(\mathbb{S}\cup \mathbbm{i})$
\end{claim}
\begin{proof}
Take any path $P$ from $\mathbb{S}\cup \mathbbm{j}$ to $[\mathbbm{n}]$
that contains a division at which $i$ is {\sc no}-decisive {\em in the modified game} $\hat{\mathcal{G}}$.
Thus $P$ contributes to $\hat{\alpha}^-_i(\mathbb{S}\cup \mathbbm{j})$. We wish to find a matching contribution to 
$\alpha^-_i(\mathbb{S}\cup \mathbbm{j}) + \alpha^-_j(\mathbb{S}\cup \mathbbm{i})$.
Let $i$ be {\sc no}-decisive at $\mathbb{T}\cup \mathbbm{j}$ on path $P$ in $\hat{\mathcal{G}}$.
We now have two cases.
\begin{enumerate}
\item $\mathbb{T}\cup \mathbbm{j}$ is losing in ${\mathcal{G}}$: 
in this case, $i$ is also {\sc no}-decisive at $\mathbb{T}\cup \mathbbm{j}$ on path $P$ in ${\mathcal{G}}$.
Thus $P$ contributes to $\alpha^-_i(\mathbb{S}\cup \mathbbm{j})$. 
\item $\mathbb{T}\cup \mathbbm{j}$ is winning in ${\mathcal{G}}$:
Now consider the {\em twin path} $P^T$ from $\mathbb{S}\cup \mathbbm{i}$ to $[\mathbbm{n}]$ which is identical
to $P$ except the roles or $i$ and $j$ are switched (note that unlike for the mirror path $P^M$ the twin path
starts at a different division than $P$). So $P^T$ passes through $\mathbb{T}\cup \mathbbm{i}$. 
But since $j$ is a {\sc yes}-blocker in ${\mathcal{G}}$ it 
must be the case that $\mathbb{T}\cup \mathbbm{i}$ is losing in ${\mathcal{G}}$. It follows that $j$ is 
{\sc no}-decisive at $\mathbb{T}\cup \mathbbm{i}$ on the path $P^T$ in ${\mathcal{G}}$.
Thus $P^M$ contributes to $\alpha^-_j(\mathbb{S}\cup \mathbbm{i})$.
Furthermore, observe that since the path $P^T$ originates at $\mathbb{S}\cup \mathbbm{i}$, the voter $i$ cannot be 
{\sc no}-decisive  on the path. 
\end{enumerate}
Ergo, the combined contribution of $P$ and $P^T$ to $\hat{\alpha}^-_i(\mathbb{S}\cup \mathbbm{j})$ is at most their combined
contribution to $\alpha^-_i(\mathbb{S}\cup \mathbbm{j}) + \alpha^-_j(\mathbb{S}\cup \mathbbm{i})$. The claim follows.
\qedhere
\end{proof}
Together Claim~\ref{cl:1} and Claim~\ref{cl:2} imply (\ref{eq:suff-1}). 
Thus the proof of Lemma~\ref{lem:crux} is complete.
\end{proof}

We are now ready to prove Theorem~\ref{thm:RM-sub}.
\begin{proof}({\em of Theorem~\ref{thm:RM-sub}})
Lemma~\ref{lem:crux} and the dummy postulate imply that the subadditivity blocker postulate holds for $I=\{i,j\}$ where
$j$ is a {\sc yes}-blocker. But if $j$ is a {\sc yes}-blocker in $\mathcal{G}$ then $I$ is a {\sc yes}-blocker in $\hat{\mathcal{G}}$.
Thus for $|I|> 2$ the result then follows by iteratively adding the voters of $I$ to the set $\{i,j\}$.
\end{proof}

\subsection*{Proofs for Section~\ref{sec:BPP}}

\restatetheorem{thm:PB-smp}
PB does not satisfy the weak minimum-power blocker postulate (and, thus, does not satisfy the strong minimum-power blocker postulate).
\end{theorem}
\begin{proof}
Take a unanimity game.
Then each player $b$ is a {\sc yes}-blocker and has voting power 
$PB_b= \sum_{\mathbb{S}\in \mathcal{D}} \alpha_i^{PB}(\mathbb{S})\cdot \gamma^{PB}(\mathbb{S})=\frac{2}{2^n}=\frac{1}{2^{n-1}}$.
Furthermore the smallest possible {\sc yes}-successful set is $S^*=[n]$, which has cardinality $n$.
A dictator $d$ has voting power $PB_d=1$.
But then, for large $n$, we have
$$PB_b 
\ =\ \frac{1}{2^{n-1}}\ \ll \ \frac{\psi_d}{|S^*|}  \ =\  \frac{1}{n}$$ 
Ergo, the weak minimum-power blocker postulate does not hold.
\end{proof}

\restatetheorem{thm:SS-smp}
SS satisfies the strong minimum-power blocker postulate (and, thus, satisfies the weak minimum-power blocker postulate).
\end{theorem}
\begin{proof}
Since $\sum_{i\in [n]} SS_i = 1= SS_d$, this follows immediately from the fact that SS satisfies the 
blocker's share postulate (Felsenthal and Machover 1998).
\end{proof}

\restatetheorem{thm:RM-smp}
RM satisfies the strong minimum-power blocker postulate (and, thus, satisfies the weak minimum-power blocker postulate).
\end{theorem}
\begin{proof}
Let $|S^*|=k$.
By unanimity, we may assume $k\ge 1$.
Let $b$ be a {\sc yes}-blocker; it immediately follows that $b\in S^*$.
Now any set of players $S$ can be written as $S=(S\cap S^*)\cup (S\cap ([n]\setminus S^*))$.
Thus we have
\begin{equation}\label{eq:blocker-1}
RM'_b
\ =\  \frac{1}{2^n} \cdot \sum_{\mathbb{S}\in \mathcal{D}} \alpha_b(\mathbb{S}) 
\ =\  \frac{1}{2^n} \cdot \sum_{S\subseteq [n]\setminus S^*}\, \sum_{T\subseteq S^*} \alpha_b(\mathbb{S}\cup \mathbb{T} ) 
\end{equation}

Now $|T|\le k$ for any subset $T$ or $S^*$. Hence
\begin{equation}\label{eq:blocker-2}
\sum_{T\subseteq S^*} \alpha_b(\mathbb{S}\cup \mathbb{T} ) 
\ =\  \sum_{\ell=0}^k\,\sum_{T\subseteq S^*: |T|=\ell} \alpha_b(\mathbb{S}\cup \mathbb{T} ) 
\ =\   \sum_{\ell=0}^k\,\sum_{T\subseteq S^*: |T|=\ell} \left( \alpha^+_b(\mathbb{S}\cup \mathbb{T} ) 
+ \alpha^-_b(\mathbb{S}\cup \mathbb{T} )\right)
\end{equation}
Now if $\mathbb{S}\cup \mathbb{T}$ is winning then $\alpha^+_b(\mathbb{S}\cup \mathbb{T})=1$
because $b$ is a {\sc yes}-blocker and, hence, is {\sc yes}-decisive in $\mathbb{S}\cup \mathbb{T}$.
It follows that to lower bound (\ref{eq:blocker-2})
we may assume that $\mathbb{S}\cup \mathbb{T}$ is losing for any $T\subset S^*$.
 Note the strict subset is necessary here since, by monotonicity, $\mathbb{S}\cup \mathbb{S}^*$ must be winning.
Thus for any $l<k$ we obtain a lower bound of
\begin{align}\label{eq:blocker-lb}
\sum_{T\subset S^*: |T|=\ell} \alpha_b(\mathbb{S}\cup \mathbb{T} ) 
&= \sum_{T\subset S^*: |T|=\ell}  \alpha^-_b(\mathbb{S}\cup \mathbb{T} ) \nonumber \\
&= \sum_{T\subset S^*: |T|=\ell, b\in T}  \alpha^-_b(\mathbb{S}\cup \mathbb{T} ) 
	+  \sum_{T\subset S^*: |T|=\ell, b\notin T}  \alpha^-_b(\mathbb{S}\cup \mathbb{T} )\nonumber \\
&= 0 +  \sum_{T\subset S^*: |T|=\ell, b\notin T}  \alpha^-_b(\mathbb{S}\cup \mathbb{T} ) \nonumber\\
&= \sum_{T\subset S^*: |T|=\ell, b\notin T}  \alpha^-_b(\mathbb{S}\cup \mathbb{T} )
\end{align}
Here the third equality holds since, by definition, the {\sc no}-efficacy of $b$ is zero for any division in which $b$ votes {\sc yes}. 
Now recall that to calculate $\alpha^-_b(\mathbb{S}\cup \mathbb{T})$ we may perform a random walk in the {\sc no}-poset.
We need to find the probability that, starting the walk at the node for division $\mathbb{S}\cup \mathbb{T}$ we reach a node 
where $b$ is {\sc no}-decisive.
In particular, $b$ is {\sc no}-decisive at the node for $\mathbb{S}\cup \mathbb{S}^*\setminus\{\mathbbm{b}\}$.
Moreover, since $b$ is a {\sc yes}-blocker, it is {\sc no}-decisive at the node for 
$\mathbb{S}\cup \mathbb{X} \cup \mathbb{S}^*\setminus\{\mathbbm{b}\}$ for any $X\subseteq [n]\setminus(S\cup S^*)$, 
by monotonicity. This implies that if we randomly add players in order to $S\cup T$ we will reach a node where $b$ is {\sc no}-decisive if
$b$ appears after every other node of $S^*$. 
Since $|T|=\ell$, this occurs with probability $\frac{1}{k-\ell}$. Thus $\alpha^-_b(\mathbb{S}\cup \mathbb{T})\ge \frac{1}{k-\ell}$.
Note this is an inequality not equality.
(We have a lower bound on {\sc no}-efficacy since there may be other losing divisions, reachable from $\mathbb{S}\cup \mathbb{T}$, 
that do not contain $S^*\setminus\{b\}$ where $b$ is {\sc no}-decisive.)
Simple counting arguments then give
 \begin{equation}\label{eq:blocker-count}
\sum_{T\subset S^*: |T|=\ell, b\notin T}  \alpha^-_b(\mathbb{S}\cup \mathbb{T} )
\ \ge\  {k-1\choose \ell}\cdot \frac{1}{k-\ell}
\ =\ {k\choose \ell}\cdot \frac{1}{k}
\end{equation}
Observe that $\alpha^+_b(\mathbb{S}\cup \mathbb{S}^* )=1$.
So, plugging (\ref{eq:blocker-lb}) and (\ref{eq:blocker-count}) into (\ref{eq:blocker-1}) gives
\begin{align*}
RM'_b &= \frac{1}{2^n} \sum_{S\subseteq [n]\setminus S^*}\, \sum_{T\subseteq S^*} \alpha_b(\mathbb{S}\cup \mathbb{T} ) \\
&= \frac{1}{2^n} \sum_{S\subseteq [n]\setminus S^*}\, \left(1 +\sum_{\ell<k} {k\choose \ell}\cdot \frac{1}{k}\right)\\
&\ge \frac{1}{2^n} \sum_{S\subseteq [n]\setminus S^*}\, \sum_{\ell=0}^k {k\choose \ell}\cdot \frac{1}{k}\\
&=  \frac{1}{k}\cdot \frac{1}{2^n} \sum_{S\subseteq [n]\setminus S^*}\, \sum_{\ell=0}^k {k\choose \ell}\\
&=  \frac{1}{k}\cdot \frac{1}{2^n} \sum_{S\subseteq [n]\setminus S^*}\, 2^k
\end{align*}
We thus obtain
\begin{equation*}
RM'_b 
\ \ge\  \frac{1}{k}\cdot \frac{2^k}{2^n} \sum_{S\subseteq [n]\setminus S^*}\, 1
\ =\  \frac{1}{k}\cdot \frac{2^k}{2^n}\cdot  2^{n-k}
\ =\  \frac{1}{k}
\ =\  \frac{\psi_d}{k}
\end{equation*}
Therefore, {\sc (smp-1)} is satisfied. A symmetrical argument applies to {\sc (smp-2)} for a {\sc no}-blocker.
Ergo, the strong minimum-power blocker postulate is satisfied.
\end{proof}

We remark that an alternative proof of Theorem~\ref{thm:RM-smp} is via the fact that RM satisfies the strong subadditivity blocker
postulate (Theorem~\ref{thm:RM-sub}). In contrast, such an approach cannot be used to prove Theorem~\ref{thm:SS-smp}
since SS does not satisfy the strong subadditivity blocker postulate (Theorem~\ref{thm:SS-sub}).

\subsection*{Proofs for Section~\ref{sec:blocker-postulates}}

\restatetheorem{thm:added-PB}
PB satisfies the added-blocker postulate.
\end{theorem}
\begin{proof}
Recall $\psi_i^+ = \psi_i^-$ for any decisiveness measure.
Thus, a decisiveness measure that satisfies Felsenthal and Machover's (1998: 266-75) specification of the added-blocker postulate ipso facto satisfies our specification.
It follows that PB satisfies our added-blocker postulate, since, as they show, it satisfies theirs.
\end{proof}

\restatetheorem{thm:added-SS}
SS does not satisfy the added-blocker postulate.
\end{theorem}
\begin{proof}
Consider the weighted voting games $\mathcal{G}=\{3; 2,1,1\}$ and $\mathcal{G}^Y=\{8; 2,1,1,5\}$ (Felsenthal and Machover 1998).
Observe that the new player in $\mathcal{G}^Y$ is a {\sc yes}-blocker.
It can be verified that $SS^+_1(\mathcal{G})=\frac{2}{6}$ and $SS^+_2(\mathcal{G})=\frac{1}{12}$ whereas 
$SS^+_1(\mathcal{G}^Y)=\frac{5}{24}$ and $SS^+_2(\mathcal{G}^Y)=\frac{1}{24}$. Thus
$\frac{SS^+_1(\mathcal{G})}{SS^+_2(\mathcal{G})} =4 < \frac{SS^+_1(\mathcal{G}^Y)}{SS^+_2(\mathcal{G}^Y)}=5$,
in violation of {\sc (add-1)} and hence our postulate.
\end{proof}

\restatetheorem{thm:added-RM}
PB satisfies the added-blocker postulate.
\end{theorem}

\noindent To prove that RM satisfies the added-blocker postulate, we show it satisfies {\sc (add-1)} and {\sc (add-2)}.
We begin with a useful lemma.
\begin{lemma}\label{lem:added-RM}
For any player $i$, the efficacy scores $\alpha^{+RM}$ and $\alpha^{-RM}$ in $\mathcal{G}$ and $\mathcal{G}^Y$ satisfy
\begin{align}
\alpha^+_{i,\mathcal{G}^Y}(\mathbb{S}\cup\{0\}) &=\alpha^+_{i,\mathcal{G}}(\mathbb{S})  &\forall S  \tag{B1} \label{eq:B1}\\
\alpha^+_{i,\mathcal{G}^Y}((S, \bar{S}\cup{0}))   &= 0 &\forall S   \tag{B2}\label{eq:B2}
\end{align}
\end{lemma}
\begin{proof}
Consider the games $\mathcal{G}$ and $\mathcal{G}^Y$.
Recall the winning divisions in $\mathcal{G}^Y$ are of the form $\mathbb{S}\cup\{0\}$ where $S\in \mathcal{W}$ in the original voting game.
The key facts are then the following. Let $S\subseteq [n]$ contain player $i$.
Then in $\mathcal{G}^Y$ player $i$ is never {\sc yes}-decisive at $(S, \bar{S}\cup{0})$ because the division is a losing 
division (given that $S$ does not contain the blocker $0$).
Furthermore, player $i$ is {\sc yes}-decisive at $\mathbb{S}\cup\{0\}$ in $\mathcal{G}^Y$ if and only if it is 
originally {\sc yes}-decisive at $\mathbb{S}$ in $\mathcal{G}$.

Next take $S\subseteq [n]$ where $S$ does not contain player $i$.
Then in $\mathcal{G}^Y$ player $i$ is never {\sc no}-decisive at 
$(S, \bar{S}\cup{0})$ because $S\cup\{i\}\notin\mathcal{W}$ (given that it does not contain the blocker~$0$).
Furthermore, player $i$ is {\sc no}-decisive at $\mathbb{S}\cup\{0\}$ in $\mathcal{G}^Y$ if and only if $i$ is 
originally {\sc no}-decisive at $\mathbb{S}$ in $\mathcal{G}$.

These key facts imply that the {\sc yes}-outcomes for $\mathcal{G}^Y$ are identical to the {\sc yes}-outcomes for $\mathcal{G}$, 
except that for each winning coalition the set of {\sc yes}-voters now also contains the blocker~$0$ (equivalently, the set of {\sc no}-voters in 
their corresponding divisions are identical).
This immediately implies that we recursively calculate the {\sc yes}-efficacy scores, they are identical for the corresponding
winning coalitions in $\mathcal{G}$ and $\mathcal{G}^Y$.
That is, $\alpha^+_{i,\mathcal{G}^Y}(\mathbb{S}\cup\{0\}) =\alpha^+_{i,\mathcal{G}}(\mathbb{S})$ and (\ref{eq:B1}) holds.

Furthermore, given that $0$ is a {\sc yes}-blocker, any $S\subseteq [n]$ is a losing division. 
Thus $\alpha^+_{i,\mathcal{G}^Y}(S, \bar{S}\cup{0})=0$ and (\ref{eq:B2}) holds.
\end{proof}

\begin{lemma}\label{lem:added-yes}
RM satisfies {\sc (add-1)}.
\end{lemma}
\begin{proof}
For any player $i$ in the original game
\begin{align*}
RM'^+_i(\mathcal{G}^Y) 
&= \sum_{\mathbb{S}\in \mathcal{D}^Y} \alpha^+_{i,\mathcal{G}^Y}(\mathbb{S})\cdot \frac{1}{2^{n+1}} \\
&= \sum_{\mathbb{S}\in \mathcal{D}} \alpha^+_{i,\mathcal{G}^Y}((S, \bar{S}\cup{0}))\cdot \frac{1}{2^{n+1}} 
	+ \sum_{\mathbb{S}\in \mathcal{D}} \alpha^+_{i,\mathcal{G}^Y}(\mathbb{S}\cup\{0\})\cdot\frac{1}{2^{n+1}}  \\
&= 0 + \sum_{\mathbb{S}\in \mathcal{D}} \alpha^+_{i,\mathcal{G}^Y}(\mathbb{S}\cup\{0\})\cdot \frac{1}{2^{n+1}}  \\
&= \sum_{\mathbb{S}\in \mathcal{D}} \alpha^+_{i,\mathcal{G}}(\mathbb{S})\cdot \frac{1}{2^{n+1}} \\
&= \frac12 \cdot \sum_{\mathbb{S}\in \mathcal{D}} \alpha^+_{i,\mathcal{G}}(\mathbb{S})\cdot \frac{1}{2^{n}} \\
&= \frac12 \cdot RM'^+_i(\mathcal{G}) 
\end{align*}
Here the third and fourth equalities hold by (\ref{eq:B1}) and (\ref{eq:B2}) of Lemma~\ref{lem:added-yes}, respectively.
Similarly, for any player $j$ in the original game, we have 
$RM'^+_j(\mathcal{G}^Y)=  \frac12 \cdot RM'^+_j(\mathcal{G})$.
Consequently,
\begin{align*}
\frac{RM'^+_i(\mathcal{G}^Y)}{RM'^+_j(\mathcal{G}^Y)}
\ =\  \frac{ \frac12 \cdot RM'^+_i(\mathcal{G})}{ \frac12 \cdot RM'^+_j(\mathcal{G})}
\ =\ \frac{RM'^+_i(\mathcal{G})}{RM'^+_j(\mathcal{G})}
\end{align*}
Ergo, the added-{\sc yes}-blocker postulate {\sc (add-1)} is satisfied.
\end{proof}

\begin{lemma}\label{lem:added-no}
RM satisfies {\sc (add-2)}.
\end{lemma}
\begin{proof}
Applying a symmetric argument to that used in the proof of Lemma~\ref{lem:added-yes} we have the following.
For any player $i$ in the original game, the efficacy scores $\alpha^+$ and $\alpha^-$ in $\mathcal{G}$ and $\mathcal{G}^N$ satisfy 
\begin{align*}
\alpha^-_{i,\mathcal{G}^N}(\mathbb{S}\cup\{0\}) &= 0 &\forall S\\
\alpha^-_{i,\mathcal{G}^N}((S, \bar{S}\cup\{0\}))  &=\alpha^-_{i,\mathcal{G}}(\mathbb{S}) &\forall S
\end{align*}
Then, applying a symmetric argument to that used in the proof of Lemma~\ref{lem:added-yes} 
completes the proof.
\end{proof}
It follows by Lemmas~\ref{lem:added-yes} and~\ref{lem:added-no}
that RM satisfies the added-blocker postulate and Theorem~\ref{thm:added-RM} holds.

\end{document}